\documentclass[copyright,envcountsame]{eptcs}
\providecommand{\event}{QAPL 2011}
\usepackage{amsmath,amssymb,amsfonts}
\usepackage{graphicx}

\newcommand{\real}{\mathbb{R}}
\newcommand{\realnn}{\real_{\geqslant 0}}
\newcommand{\realp}{\real_{> 0}}
\newcommand{\nat}{\mathbb{N}}
\newcommand{\Tt}{\mathcal{T}}
\newcommand{\Aa}{\mathcal{A}}

\newcommand{\Bb}{\mathcal{B}}
\newcommand{\Vv}{\mathcal{V}}
\newcommand{\CostFun}{\mathcal{CF}}
\newcommand{\set}[1]{\{ #1 \}}

\newcommand{\eset}[1]{\set{#1}}
\newcommand{\struct}[1]{\langle #1 \rangle}
\newcommand{\sq}{\::\:}

\newcommand{\move}[1]{\xrightarrow{#1}}
\newcommand{\piecewise}[1]{\langle #1 \rangle}
\DeclareMathOperator{\minC}{minC}
\DeclareMathOperator{\maxC}{maxC}
\DeclareMathOperator{\argmin}{arg\,min}

\DeclareMathOperator{\override}{\triangleright}

\newcommand{\lmin}{l^*}
\newcommand{\mMIN}{\mathrm{Min}}
\newcommand{\mMAX}{\mathrm{Max}}
\newcommand{\mGoal}{\mathrm{Goal}}
\newcommand{\Loc}{L}
\newcommand{\LGoal}{\Loc^{\mGoal}}
\newcommand{\LMax}{\Loc^{\mMAX}}
\newcommand{\LMin}{\Loc^{\mMIN}}
\newcommand{\urg}{mathrm{urg}}
\newcommand{\fgoal}{f^{\mGoal}}
\newcommand{\afgoal}{{mathrm{CF}^\mGoal}}
\newcommand{\Length}{mathrm{Len}}
\newcommand{\Last}{mathrm{Last}}
\newcommand{\Runs}{\mathrm{Runs}}
\newcommand{\FRuns}{\Runs_{\mathrm{fin}}}
\newcommand{\Run}{\mathrm{Run}}
\newcommand{\Stop}{\mathrm{Stop}}
\newcommand{\Cost}{\mathrm{Cost}}
\newcommand{\Price}{\mathrm{Price}}
\newcommand{\States}{S}
\newcommand{\SGoal}{\States^{\mGoal}}
\newcommand{\SMin}{\States^{\mMIN}}
\newcommand{\SMax}{\States^{\mMAX}}
\newcommand{\SigmaMin}{\Sigma^{\mMIN}}
\newcommand{\SigmaMax}{\Sigma^{\mMAX}}
\newcommand{\PiMin}{\Pi^{\mMIN}}

\newcommand{\OptCost}{\mathrm{OptCost}}
\newcommand{\nonurg}{\mathrm{NonUrgent}}
\newcommand{\solverp}{\mathtt{SolveRP}}
\newcommand{\tcompose}{\mathtt{CostConsistent}}

\newtheorem{theorem}{Theorem}

\newtheorem{definition}[theorem]{Definition}
\newtheorem{remark}[theorem]{Remark}
\newtheorem{lemma}[theorem]{Lemma}
\newtheorem{example}[theorem]{Example}
\newtheorem{problem}[theorem]{Problem}
\newtheorem{proposition}[theorem]{Proposition}
\newcommand{\qed}{\hfill$\square$}

\newenvironment{proof}{\noindent\emph{Proof.}\hspace{0.5em}}{}

\title{Two-Player Reachability-Price Games on Single-Clock Timed
  Automata} \def\titlerunning{Two-Player Reachability-Price Games on
  Single-Clock Timed Automata} \author{Micha{\l} Rutkowski
  \institute{Department of Computer Science, University of Warwick,
    UK}}\def\authorrunning{Micha{\l} Rutkowski, University of Warwick}

\begin{document}

\maketitle

\begin{abstract}
   We study two player reachability-price games on single-clock timed
   automata. The problem is as follows: given a state of the
   automaton, determine whether the first player can guarantee
   reaching one of the designated goal locations. If a goal location
   can be reached then we also want to compute the optimum price of
   doing so. Our contribution is twofold. First, we develop a theory
   of cost functions, which provide a comprehensive methodology for the analysis of
   this problem. This theory allows us to establish our second
   contribution, an EXPTIME algorithm for computing the optimum
   reachability price, which improves the existing 3EXPTIME upper
   bound. 
\end{abstract}

%
%
\section{Introduction}
Timed automata~\cite{AD94} are a formalism used for modeling real time
systems, i.e., systems whose behavior depends on time. Timed automata
are finite automata augmented with a set of clocks. The values of the
clocks grow uniformly over time. There are two types of transitions:
continuous, resulting in time progression, and discrete, resulting in
a change of location. Discrete transitions may reset values of certain
clocks to zero, and different transitions may be enabled at different
clock values.

Optimal schedule synthesis is one of the key areas of research in
timed automata theory \cite{BBL04,BBL08,BLMR06,JT07,JT08}. In this
setting, timed automata are augmented with pricing information, and
each execution of the automaton is assigned a payoff. Moreover, we
want to model lack of full control over the system; game theory is
commonly used in this context
\cite{BBJLR08,BLMR06,JLR09,JT08,JT07}. There are two players: the
minimizer and the maximizer\footnote{In the literature, these players
  are often referred to as the controller and the environment.}, who
have opposite goals of minimizing and maximizing the payoff of a play,
respectively. In this case a play is an execution of the automaton,
and we are dealing with the worst case scenario, where the controller
is interacting with an adversarial environment. 
In this context, reachability-price games are commonly
considered~\cite{BBM06,BLMR06,JT07}.  In these games the goal is to
optimize the accumulated price of reaching a designated set of states.

When the payoff of an execution is simply its time duration,
synthesizing an almost-optimal schedule is EXPTIME-complete
\cite{JT07}. In linearly priced timed automata \cite{BBM06,BLMR06},
the price of an individual continuous transition is its duration
multiplied by a location specific price rate. Bouyer et al. show that
determining the existence of an optimal schedule for linearly priced
timed automata, with at least three clocks is undecidable
\cite{BBM06}. On the other hand, Bouyer et al. show a triply
exponential algorithm for single-clock linearly priced timed automata
\cite{BLMR06}. However, the exact complexity of the problem is still
unknown, as PTIME is the best lower bound that is currently known.

\paragraph{Contributions.} 
In this paper, we present a new EXPTIME algorithm for optimal schedule
synthesis for linearly priced single-clock timed automata with
non-negative price rates. Our work improves the triply exponential
algorithm given by Bouyer et al. \cite{BLMR06}.

Our contribution is twofold. First, in order to deliver the main
result, we establish technical results regarding cost functions, i.e.,
piecewise affine continuous functions that are non-increasing. Cost
functions in this form were first considered by Bouyer et
al. \cite{BLMR06}. They are central to both algorithms, the one
presented in this paper, and that of Bouyer et al. \cite{BLMR06}, as
both algorithms use them to produce their output. The output of the
algorithms is a function that assigns the optimal price of
reachability to each state; this output function can be represented by a
finite set of cost functions. Our technical results regard operations
performed on cost functions during the execution of the algorithm. We
establish the properties and invariants of these operations, which
later allows us to analyze the complexity, and prove the correctness
of the algorithm. This understanding of cost functions was pivotal in
achieving the doubly-exponential speedup, and we believe that this
detailed analysis might prove useful in further closing the existing
complexity gap.

Second, we show an EXPTIME algorithm for computing the optimal price
of reachability. As in Bouyer's et al. approach \cite{BLMR06}, the
algorithm does the computation through a recursive procedure, with
respect to the number of locations of the timed
automaton. Again, as in Bouyer's et al. work \cite{BLMR06}, in each
recursive call we single out the location that minimizes the price
rate. In the model, locations are assigned to players, which
necessitates different handling of a location, depending on its
ownership. In the case of the maximizer locations, our algorithm
behaves exactly like the original, however, when it comes to handling
minimizer locations, we improve over the predecessor. The original
algorithm would proceed to recursively solve two subproblems, which
resulted in an additional exponential blowup. The algorithm presented
in this paper, as in the case of maximizer locations, employs an
iterative procedure which prevents this blowup. The approach is
similar in spirit to that used in handling maximizer locations,
however, the details are different.

Tools for handling games, where the price of an execution is its time
duration, already exist (e.g., UPPAAL~\cite{BDL04}). We believe that the work
presented in this paper may help in the development of such tools for
reachability-price games on linearly priced timed automata.

%
%
\section{Preliminaries}
\label{sec:preliminaries}
\paragraph{Cost functions.}
Below we introduce the notion of a cost function, and prove some of
its basic properties. Cost functions are a central notion when
considering reachability-price games on single-clock timed automata. The
theory of cost functions will be used to construct the algorithm for
computing the optimal reachability cost, as well as to prove its
correctness.

In this paper, we will be dealing with the ordered set of real numbers
augmented with the greatest element, positive infinity. For that
purpose we need to extend the $+$, $\min$ and $\max$ operators in a
natural way. For $a \in \real \cup \eset{\infty}$ we have $a+\infty =
\infty$, $\max(a,\infty) = \infty$ and $\min(a,\infty) = a$.

\begin{definition}[Cost function]
  A function $f : I \to \real \cup \eset{\infty}$ that is continuous,
  non-increasing, and piecewise affine, where $I$ is a bounded
  interval, is said to be a \emph{cost function}.  We will write
  $\CostFun(I) \subseteq [I\to\real \cup \eset{\infty}]$ to denote the
  set of all cost functions with the domain $I$.
\end{definition}

\begin{remark}
  Notice that if $f \in \CostFun(I)$, and $f(x) = \infty$ for some $x
  \in I$ then $f \equiv \infty$, over $I$.\qed
\end{remark}

At times, we will need to talk about the individual affine functions,
i.e., the pieces of a cost function. To make this easier, we introduce
the following convention.  Given an interval $I$, let $f : I \to
\real$ be a cost function, we will write $f =
\piecewise{f_1,\ldots,f_k}$ to denote the fact that the piecewise
affine function~$f$ consists of affine pieces $f_1,\ldots,f_k$, with
domains $I_1,\ldots,I_k$, where $k$ is the smallest integer such that
\[
  f(x) = \begin{cases}
    f_1(x) & x \in I^f_1\\
    \multicolumn{2}{c}{\dotsb}\\
    f_k(x) & x \in I^f_k.
    \end{cases} 
\]
Throughout the paper, we will be implicitly assuming that $I = [b,e]$,
and that $I^f_i = \left[b_i^f,e_i^f\right]$, for $i \in
\eset{1,\ldots,k}$, with $e_{i+1}^f = b_i^f$, for $i \in
\eset{1,\ldots,k-1}$. The formula for the individual segment $f_i$
will be given by $a^f_i\cdot x + c^f_i$, for $i \in
\eset{1,\ldots,k}$. If $f$ is clear from the context, we will omit the
superscript.

We now introduce two operators, which are key in defining the
relationship between reachability cost functions of a location and its
successors. This will be later summarized by Lemma
\ref{lem:OptCost_opt_eq}.  Given a cost function $f : [b,e] \to
\real\cup \eset{\infty}$ and a positive constant $c$, we define the
following two operators, that transform cost functions:
\[
  \minC(f,c) = x \mapsto \min_{0 \leqslant t \leqslant e-x} ct + f(x+t)
  \]
and $\maxC(f,c)$ defined analogously, with max substituted for min.

\begin{lemma}
  \label{lemma:minC_piecewise_affine}
  Let $c$ be a positive constant. If $f : [b,e] \to \real$ is a cost
  function then $\minC(f,c)$ is a cost functions as well. 
  The same holds for $\maxC$.
 \end{lemma}

The following Proposition formalizes the intuition how the $\minC$
($\maxC$) operators affect the function $f$. The $\minC$ ($\maxC$)
operator removes all pieces of $f$ that have slopes steeper
(shallower) than $-c$, and substitutes them with pieces that have a
slope equal to~$-c$. For the remaining pieces, the formula remains
unchanged, but the domain may change. However, the new domain is
always a subset of the domain in $f$.

\begin{proposition}
  \label{prop:minC_segment_preservation}
  Let $f = \piecewise{f_1,\ldots,f_k}$ and let $\minC(f,c) =
  \piecewise{g_1,\ldots,g_l}$.  We have that $l \leqslant k$, and for every
  $j \leqslant l$: if the formulas for $g_j$ and $f_i$ are equal, for
  some $i \leqslant k$, then $b^f_i \in I^g_j \subseteq I^f_i$,
  otherwise $a^g_j > -c$. Moreover, $a^g_j \geqslant -c$, for all
  $j=1,\ldots,l$.
\end{proposition}

For $\maxC$ we can prove a similar result, with the only difference
that in the statement of Prop.~\ref{prop:minC_segment_preservation}
$<$ and $\leqslant$ are substituted for $>$ and $\geqslant$.

\begin{example}
  \label{ex:minC_operator}
  Fig.~\ref{fig:ta_minc_operator_scheme} gives an intuitive
  understanding of the $\minC(f,c)$ operator, for a cost function $f$
  and a positive real constant $c$. The cost function is given as $f =
  \piecewise{f_1,\ldots,f_7}$ (as seen in
  Fig.~\ref{fig:ta_minc_operator_scheme}~a)). The slope of $f_2$ and
  $f_6$ is smaller than $-c$; for the remaining components it is
  greater. The cost function $g = \minC(f,c)$ is depicted in
  Fig.~\ref{fig:ta_minc_operator_scheme}~b), and is given by
  $\piecewise{g_1,\ldots, g_6}$. All components of $g$ have a slope greater
  or equal to $-c$. The formula for $g_1$ is the same as for $f_1$,
  however, the domain is a subset (similarly for $f_4$ and $g_4$). The
  function $g_3$ is equal to $f_3$ (similarly $g_6$ is equal to
  $f_7$). Functions $g_2$ and $g_5$, have the slope $-c$, and where
  not present in $f$.\qed
\end{example}

\begin{figure}
  \begin{center}
  \begin{tabular}{c p{2em}c}
    \includegraphics[scale=0.3]{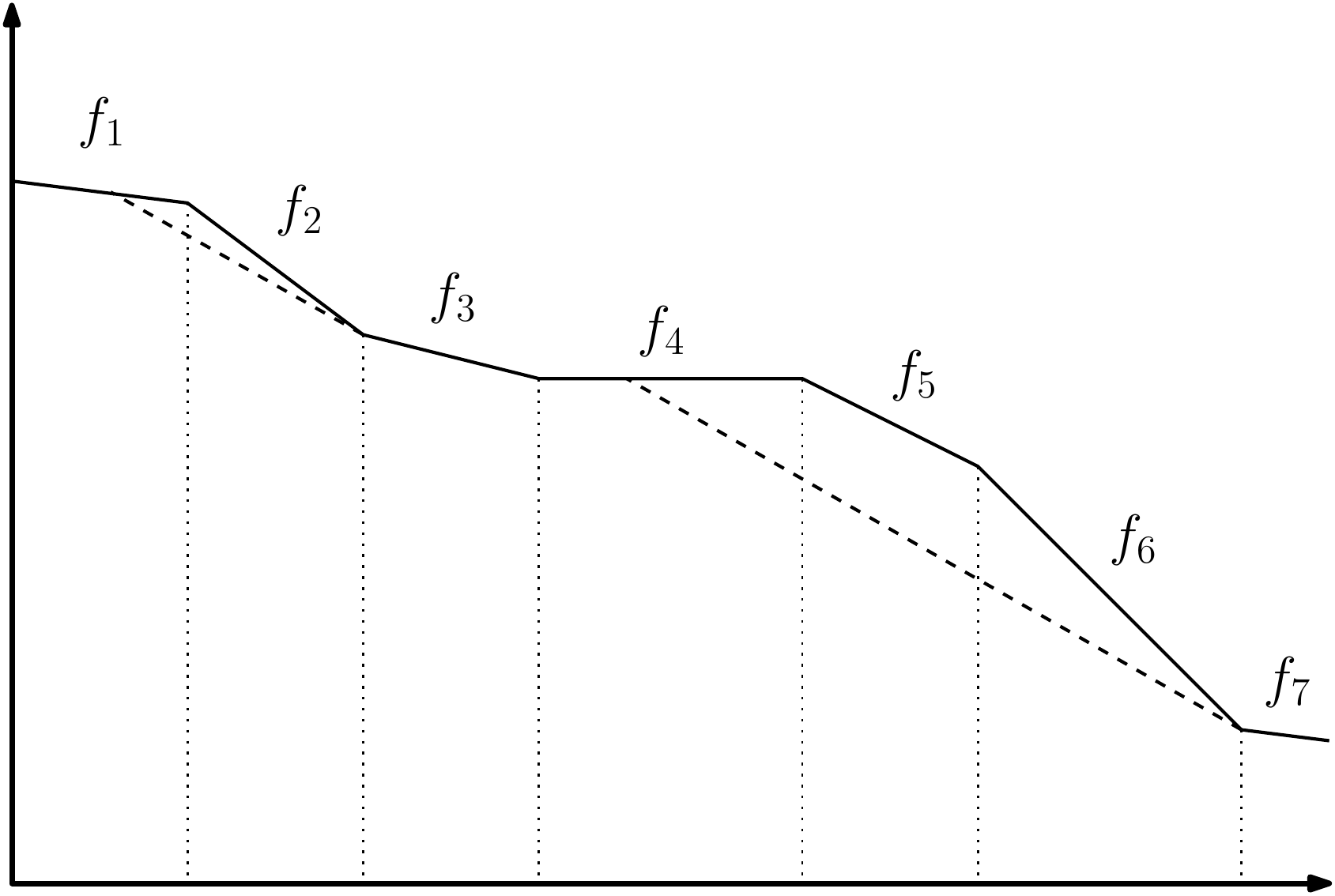} &&
    \includegraphics[scale=0.3]{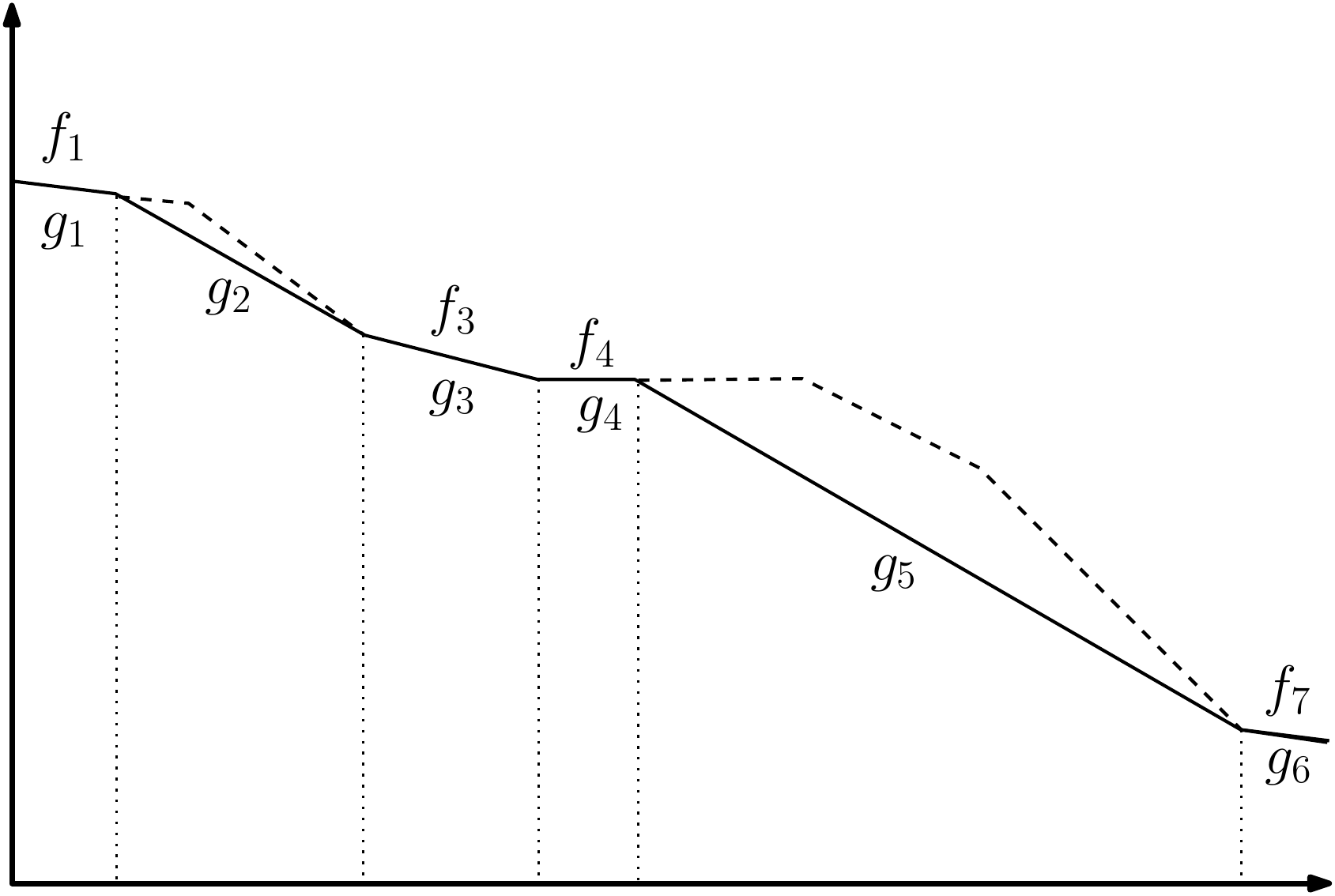}\\
    a) && b)
  \end{tabular}
  \end{center}
  \caption{a) Cost function $f$, before applying $\minC(f,c)$ operator
    --- dashed lines have a slope $-c$; b) cost function $g =
    \minC(f,c)$ --- dashed lines denote parts of $f$ that do not
    coincide with $g$.}
  \label{fig:ta_minc_operator_scheme}
\end{figure}

\paragraph{Reachability-price games.}
A reachability-price game is played on a transition system whose
states are partitioned between two players, the minimizer and the
maximizer. A game starts in some state, and the players change the
current state according to the transition rules, with the owner of the
state deciding which transition to take. The goal of the minimizer is
to reach a state in the designated set of goal states, whereas the
goal of the maximizer is to prevent this from happening. Each
transition incurs a price, and the minimizer, if she can assure that a
goal state is reached, wants to minimise the total price of doing
so. If the maximizer cannot prevent the minimizer from reaching a
goal state, then his goal is to maximise the total price of reaching
one.

A \emph{weighted labeled transition system}, or simply a
\emph{transition system}, $\Tt = \struct{S,\Lambda,\move{},\pi}$,
consists of a set of states, $S$, a set of labels, $\Lambda$, a
labeled transition relation $\move{}\;\subseteq S\times\Lambda\times
S$, and a price function $\pi$ that assigns a real number to every
transition. We will write $s\move{\lambda}s'$ to denote a transition,
i.e., an element $(s,\lambda,s')\in\;\move{}$. We will say that the
transition system is \emph{deterministic}, if the relation $\move{}$
can be viewed as a function $\move{}\; : S \times \Lambda \to S$.

\begin{remark}
  \label{rem:transitions_systems_deterministic}
  For the remainder of the paper we will be considering deterministic
  weighted labeled transition systems.\qed
\end{remark}

Weighted transition systems will be used to provide the semantics for
single-clock timed automata, considered in this paper. In this
context, the restriction made in
Rem.~\ref{rem:transitions_systems_deterministic} is not constraining,
as single-clock timed automata yield deterministic transition
systems. We place this restriction because it allows for simpler
definitions (e.g., we rely on this restriction when defining the
notion of a run induced by the players strategies).

A \emph{reachability-price game} $\Gamma =
\struct{\Tt,\SMin,\SMax,\SGoal,\fgoal}$ consists of a weighted
transition labeled system, $\Tt$, a partition of the transition system's
set of states, into the minimizer and maximizer states, $\SMax =
S\setminus \SMin$, a designated set of goal states, $\SGoal \subseteq
S$, and goal cost function, $\fgoal : \SGoal \to \real$.

Given a state $s$, a run of the game from $s$ is a (possibly
infinite) sequence of transitions $\omega = s_0 \move{\lambda_1} s_1
\move{\lambda_2} s_2 \cdots$, where $s = s_0$. If two runs $\omega =
s_0 \move{\lambda_1}\cdots\move{\lambda_k} s_k$ and $\omega' =
s_0'\move{\lambda_1}\cdots$ are such that $s_k = s_0'$ then,
$\omega\omega'$ denotes the run $s_0
\move{\lambda_1}\cdots\move{\lambda_k} s_k \move{\lambda_1'}\cdots$.
Given a finite run $\omega$, $\Length(\omega)$ will denote its length,
i.e., the total number of transitions, $\Last(\omega)$ will denote the
final state of $\omega$, i.e., $s_{\Length(\omega)}$, and $\omega_n$
will denote the prefix of $\omega$ of length $n$, where $n \leqslant
\Length(\omega)$.  The set of all runs of $\Gamma$ is denoted by
$\Runs$. The set of all finite runs of $\Gamma$ is denoted by
$\FRuns$. Note that $\FRuns \subseteq \Runs$. We will also write
$\Runs(s)$ ($\FRuns(s)$) to denote the set of all runs (all finite
runs) starting in a state~$s$.

A strategy of the minimizer is a partial function $\mu : \FRuns \to
\Lambda$ such that for every finite run $\omega$, ending in a state of
the minimizer, $Last(\omega) \move{\mu(\omega)} s'$. We will say that
$\mu$ is positional if it can be treated as a function $\mu : S\to
\Lambda$. We will write $\SigmaMin$ and $\PiMin$ to denote the sets of
all and all positional strategies of the minimizer, respectively. The
set of strategies for the maximizer is defined analogously.

Given a run $\omega'$ ending in a state $s_0$, and a pair of
strategies $\sigma \in \SigmaMin$ and $\chi \in \SigmaMax$, we write
$\Run(\omega',\mu,\chi)$ to denote the unique run $\omega \in
\Runs(s_0)$ satisfying: if $s_{i}\move{\lambda_{i+1}}s_{i+1}$ is the
$(i+1)$-th transition, of $\omega$, then $\mu(\omega'\omega_i) =
\lambda_i$ if $s_i \in \SMin$, otherwise, $\chi(\omega'\omega_i) =
\lambda_i$. Note that, if $\mu$ and $\chi$ are positional, then
$\omega'$ is irrelevant.

Given a finite run $\omega$ we define its price, $\Price(\omega)$, as
$\sum_{i=1}^{\Length(\omega)}\pi((s_{i-1},\lambda_i,s_i))$, i.e., the
total price of its transitions. Given a run $\omega \in \Runs$, let
$\Stop(\omega) = \min\set{i \sq s_i \in \SGoal}$. The cost of a run
$\omega$ is defined as: 
\begin{displaymath}
  \Cost(\omega) = \begin{cases} \fgoal(\Last(\omega_{\Stop(\omega)}))
    + \Price(\omega_{\Stop(\omega)}) & \Stop(\omega) < \infty,\\ \infty
    & \text{ otherwise.}
    \end{cases}
\end{displaymath}

We now define the function $\OptCost : S \to \real \cup
\eset{\infty}$, which maps every state to the minimum cost of reaching
a goal state that can be guaranteed by the minimizer. If the maximizer
can prevent the minimizer from achieving a goal state, the cost is
$\infty$. The function is defined as:
\[
  \OptCost(s) = \inf_{\mu \in \SigmaMin}\sup_{\chi \in
    \SigmaMax}\Cost(\Run(s,\mu,\chi)).
\]

Finally, we introduce the notion of $\varepsilon$-optimality, for
$\varepsilon > 0$. We say that $\mu \in \SigmaMin$ is
\emph{$\varepsilon$-optimal}, if $\sup_{\chi \in
  \SigmaMax}\Cost(\Run(\omega,\mu,\chi)) \leqslant
\OptCost(\Last(\omega)) + \varepsilon$ for all $\omega \in
\FRuns$. Given a strategy $\mu \in \SigmaMin$ we say that $\chi \in
\SigmaMax$ is \emph{$\varepsilon$-optimal for $\mu$}, if
$\Cost(\Run(\omega,\mu,\chi)) \geqslant \OptCost(\Last(\omega)) -
\varepsilon$, for all $\omega \in \FRuns$.

The decision problem associated with reachability-price games is the
following:

\begin{problem}
  \label{prob:decision_problem_reachability_price}
  Given a reachability-price game $\Gamma$, its state $s$, and a real
  constant $c$, determine whether $\OptCost(s) \leqslant c$.\qed
\end{problem}

If $\Tt$ or $\Gamma$ are not clear from the context we will write
$\OptCost_\Gamma$, $\Runs_\Gamma$, etc.

\paragraph{Single-clock timed automata.}
In this paper we are considering timed automata with a single
clock. We write $X = \eset{x}$ to denote the set containing the single
clock $x$. A clock constraint is given by a closed interval with
non-negative integer end points. We write $\Bb(X)$ to denote the set
of all clock constraints.  A clock valuation is a function that
assigns a non-negative real value to the clock $x$; $\Vv = [X \to
  \realnn]$ denotes the set of all single clock valuations.  A clock
valuation $v$ satisfies a clock constraint $g \in \Bb(X)$ if $v(x) \in
g$, and this will be denoted by $v\models g$. We write $v_0$ to denote
the $x \mapsto 0$ valuation. For a valuation $v$ and $t \in \realnn$
the valuation $v+t$ denotes the valuation $x \mapsto v(x) + t$.

A \emph{weighted single-clock timed automaton} $\Aa =
\struct{L,E,\eta,\urg,\pi}$ consists of a finite set of locations,
$L$, an edge relation, $E \subseteq L \times \Bb(X)
\times 2^X \times L$, an invariant specification, $\eta : L \to
\Bb(X)$, an urgency mapping, $\urg : L \to \eset{0,1}$, and weight
function, $\pi : L \cup E \to \nat$.

We assume (without loss of generality~\cite{BFH+01}) that $\Aa$ is
clock-bounded, i.e., there exists a positive constant $M$ such that
$v \models \eta(l)$ implies $v(x) \leqslant M$, for every location $l$.

The size of the automaton, denoted by $|\Aa|$, is the total number of
bits needed to represent all of its components --- constants are
encoded in binary.

The semantics of a timed automaton $\Aa$ is given in terms of a
deterministic weighted labeled transition system $\Tt_\Aa =
\struct{S_\Aa, \Lambda_\Aa, \move{}_\Aa, \pi_\Aa}$. The set of states
$S_\Aa \subseteq L \times \Vv$ is such that $v\models\eta(l)$ for
every $(l,v) \in S_\Aa$. The set of labels is given by $\Lambda_\Aa =
E \cup \realp$. The transition relation, $\move{}_\Aa$ admits a
transition $(l,v) \move{\lambda} (l',v')$ iff one of the following is
true:
\begin{description}
  \item[Discrete transition] $\lambda = (l,g,Z,l') \in E$, $v \models
    g$, and if $Z = \emptyset$ then $v=v'$, otherwise $v'=v_0$.
  \item[Continuous transition] $\lambda = t \in \realp$, $\urg(l) = 0$,
    i.e., the location is non-urgent, for every $t' \in (0,t)$ we have
    $v+t' \models \eta(l)$, $l=l'$, and $v' = v+t$.
\end{description}
Finally, the price function, $\pi_\Aa((l,v)\move{\lambda}(l',v'))$ is
defined as $\pi(\lambda)$ if $\lambda \in E$, and $\pi(l)\cdot\lambda$,
otherwise.

We will often abuse notation, and treat the state of the automaton as
an element of $L \times \realnn$, and the clock valuation as a real
variable.

\begin{remark}
  We only allow runs that do not admit infinitely many consecutive
  continuous transitions. Note that this requirement does not exclude
  Zeno runs, i.e., infinite runs whose total duration is finite.\qed
\end{remark}

\paragraph{Reachability-price games on single-clock timed automata.}
Fix a partition of the set of locations, $L = \LMin\setminus \LMax$,
into the minimizer and maximizer locations, the set of goal locations
$\LGoal \subseteq L$, and a function that assigns a cost function to
every goal location, $\afgoal : \LGoal\to\CostFun([0,M])$, where $M$
is the clock bound. We can define a reachability-price game on a
single-clock timed automaton $\Aa$, by defining a reachability-price
game on its transition system $\Tt_\Aa$. The reachability-price game
$\Gamma_\Aa$ is given by $\struct{\Tt_\Aa,\SMin,\SMax,\SGoal,\fgoal}$,
where: $\SMin = S \cap (\LMin \times \Vv)$, $\SMax = S\setminus
\SMin$, $\SGoal = S \cap (\LGoal \times \Vv)$, and $\fgoal((l,x)) =
\afgoal(l)(x)$ for every state $(l,x) \in \SGoal$.

The size of the game, denoted by $|\Gamma|$, is the total number
of bits needed to represent all of its components --- constants are
encoded in binary.

\paragraph{Assumptions.}
We are going to place some restrictions on the structure of timed
automata, which will allow us to concentrate on the essence of the
problem. In their work, Bouyer et al. place the same restrictions, and
argue that this is without loss of generality~\cite{BLMR06}. In
particular, their complexity results is stated only for the restricted
automata.

Consider an interval $I$, we will write $\Aa_I$, for some
\emph{$I$-bounded timed automaton}, i.e., an automaton whose
transition system has the state space restricted to $L \times I$, and
for every $l$, the invariant is $\eta(l)\cap I$. To obtain the
classical automaton we need to take $I = [0,\infty]$.

We say that an automaton $\Aa$ is \emph{simple} if it is
$[0,1]$-bounded and for every discrete transition of its transition
system, the reseting set is empty, i.e., for every $e = (l,g,Z,l') \in
E$ we have $Z = \emptyset$. Notice that in simple timed-automata time
always progresses.

We have the following result regarding simple single-clock timed automata.

\begin{theorem}
  Problem~\ref{prob:decision_problem_reachability_price} for
  reachability-price games on single-clock timed automata is
  polynomially Turing reducible to the analogous problem on simple
  single-clock timed automata.
\end{theorem}

We simplify the automaton further, by assuming that the price of every
discrete transition is $0$. This assumption allows for a clearer
exposition, and is without loss of generality~\cite{BLMR06}.  A
technique similar to that used to remove the resets can be employed.
For simplicity, let $c \in \nat$ be the constant used in
Problem~\ref{prob:decision_problem_reachability_price}. For every
state $s$, we need to consider at most $c$ copies of the slightly
modified game $\Gamma$, which is played on a simple automaton with no
prices on discrete transitions. Intuitively, the $\OptCost$ function
for the $i$-th copy, gives the optimal cost of reaching goal, provided
that at most $i$ transitions with non-zero prices were executed. The
$\OptCost$ function computed for the $i$-th copy is used to construct
the $(i+1)$-th copy. Each copy is treated independently, and although
we might have to consider exponentially many, this does not increase
the complexity as our algorithm is in EXPTIME.

Simple timed automata admit three possible edge guards, namely
$[0,0]$, $[1,1]$, and $[0,1]$ (recall that we are considering only
closed intervals as clock constraints). The first kind does not allow
for a continuous transition, prior to a discrete one, and it is
satisfied only by finitely many states. As it will be visible in the
proofs of Sec.~\ref{sec:results}, the value of the $\OptCost$ function
for such states, due to the time progression property of simple timed
automata, does not ``affect'' the values for the other
states. Transitions with guards of this kind can be dealt with, in
polynomial time, during post-processing. The effect of a discrete
transition, featuring a guard of the second kind, can be encoded using
additional goal cost functions. Once again, proofs in
Sec.~\ref{sec:results} explain how this can be done.  It is only the
third kind of guards that cannot be dealt with by such simple
means. In light of this, and to simplify the presentation, we assume
that all transition guards are true. A similar approach was used in
the work of Bouyer et al.~\cite{BLMR06}.

\begin{remark}
  \label{rem:ta_transitions_have_no_guards}
  In the light of the assumptions made, it is natural to think of $E$
  as a subset of $L\times L$.\qed
\end{remark}

We will also assume that from every state the cost of reaching a goal
state is finite. In light of
Rem~\ref{rem:ta_transitions_have_no_guards}, one can determine the set
of states, from which the maximizer can prevent reaching goal, by
determining the appropriate set of locations; this can be done in
polynomial time. The real complexity lies in determining the optimal
cost of reaching a goal state, given that the minimizer can ensure it.

\paragraph{Operations.}
We will now define some simple algebraic operations that we will be
performing on cost functions, and reachability-price games on simple
timed automata.  These operations will be used in the algorithm,
presented in Sec.~\ref{sec:results}.

Given two functions $h : I_1 \to \real$ and $g : I_2\to\real$, we will
write $f \override g$ to denote the \emph{override} operation on these
two functions~\cite{BJSV10}, defined as $(h \override g)(x) = h(x)$ if $x
\in I_1$ and $g(x)$ if $x \in I_2 \setminus I_1$.

Fix an interval $I \subseteq [0,1]$, an automaton $\Aa$, and a
reachability-price game $\Gamma$, on $\Aa$. Below we list three
operations, that given a game $\Gamma$, produce a new game:
\begin{description}
  \item[] $\Gamma[\urg(l):=1]$ denotes the game $\Gamma'$ obtained from $\Gamma$
    by changing the urgency mapping of $\Aa$ so that $l$ is an urgent location.
  \item[] $\Gamma[\LGoal \cup l,h]$ denotes the game obtained from
    $\Gamma$ by adding $l$ to the set of goal locations, with $h$
    being the cost function assigned to $l$. It gives the game
    $\Gamma'$, obtained from $\Gamma$, by setting ${\LGoal}'=\LGoal
    \cup \eset{l}$, and defining the mapping from goal locations to
    cost functions, $\afgoal'$, as $(l\mapsto h)\override \afgoal$.
    Function $h : I \to \real$ is a cost function, and $l$ is a
    location.  We do not require $l \in L$, i.e., $l$ can be a fresh
    location.
  \item[] $\Gamma[E \cup e]$ the game obtained from $\Gamma$ by adding
    an additional edge $e$ in the automaton $\Aa$. The new edge set is
    equal to $E \cup \eset{e}$, where $e \in L \times L$.
\end{description}

%
%
\section{Results}
\label{sec:results}
We are interested in solving reachability-price games algorithmically.
To solve a reachability-price game $\Gamma$ means to compute the
$\OptCost$ function. In this section we present an algorithm
for computing this function. We start by introducing some preliminary
notions, then we present the algorithm, and to conclude this section we
provide a proof of its correctness.
The algorithm extends the work of Bouyer et al. \cite{BLMR06}. With
each recursive call, it attempts to solve a game with one less
non-urgent location. The problem is polynomial-time solvable, when
only urgent locations are present~\cite{BLMR06}.

In the following we will be considering a game $\Gamma$, and games
derived from it, $\Gamma'$ and $\Gamma''$. Furthermore, due to the
iterative nature of our algorithm, we will often restrict the game to
an interval, $I$. To ensure clarity, we will be writing
$\OptCost_{\Gamma_I}$ to explicitly indicate the game $\Gamma$ and the
interval $I$, to which the function refers. Unlike clock constraints,
the interval $I$ will usually have rational endpoints.  

At times, it will be convenient to treat $\OptCost$ as an element of
$[\Loc\to\CostFun(I)]$, rather than an element of $[\States \to
  \real]$. We will therefore abuse the notation, and write
$\OptCost(l)$ to denote the function $x\mapsto \OptCost(l,x)$.

To make handling of non-urgent locations easy, we introduce the
following definition:
\[
  \nonurg(\Gamma) = \set{l\sq l \in\Loc\setminus\LGoal \text{ and }
    \urg(l) = 0}
  \]

Fix a game $\Gamma$ and two intervals $I_1 = [b_1,e_1],I_2 =
[b_2,e_2]$ such that $r = e_1 = b_2$. We would like to have a way of
computing $\OptCost_{\Gamma_{I_1 \cup I_2}}$, provided that we have
already computed $\OptCost_{\Gamma_{I_2}}$. To enable this we define
the following operation:
\begin{multline*}
  \tcompose(\Gamma_{I_1},\OptCost_{\Gamma_{I_2}}) =\\ \Big(\Gamma\big[\LGoal
    \cup l_{1}',\;x\mapsto
    \OptCost_{\Gamma_{I_2}}(l_1,r)+(r-x)\pi(l_1)\big]\big[E \cup
    (l_1,l_{1}')\big]\ldots\\ \big[\LGoal \cup l_{k}',\;x\mapsto
    \OptCost_{\Gamma_{I_2}}(l_k,r)+(r-x)\pi(l_k)\big]\big[E \cup
    (l_k,l_{k}')\big]\Big)_{I_1},
\end{multline*}
where $\eset{l_1,\ldots,l_k} = \nonurg(\Gamma)$.

The intuition behind the $\tcompose$ operation is as follows.  In
$\Gamma_{I_1}$, the time cannot progress past $e_1$, whereas in
$\Gamma_{I_1 \cup I_2}$ it can; this results in
$\OptCost_{\Gamma_{I_1}}$ being unrelated to
$(\OptCost_{\Gamma_{I_1\cup I_2}})_{| \Loc \times I_{I_1}}$, although,
due to the lack of resets, $\OptCost_{\Gamma_{I_2}}$ is equal to
$(\OptCost_{\Gamma_{I_1\cup I_2}})_{| \Loc \times I_{I_2}}$. To
alleviate this, for every non-urgent location $l$, we add a new goal
location $l'$ whose cost functions encodes the following behavior:
upon entering $l$ wait until time $e_1$, and then reach goal, from the
state $(l,e_1)$, ``optimally'' as if $\Gamma_{I_2}$ was the game being
played. This intuition is formalized by the following lemma.

\begin{lemma}
  \label{lem:timed_composition}
  If $\Gamma_{I_1}' = \tcompose(\Gamma_{I_1},\OptCost_{\Gamma_{I_2}})$
  then
  \[
    \OptCost_{\Gamma_{I_1}'} \override \OptCost_{\Gamma_{I_2}} =
    \OptCost_{\Gamma_{I_1 \cup I_2}}.
    \]
\end{lemma}

We will also be considering situations where we have already computed
$\OptCost_\Gamma(l)$ for some location $l$ of $\Gamma$,
and we will want to use this fact to compute $\OptCost_\Gamma$ for the
remaining locations.

\begin{lemma}
  \label{lem:one_location_computed}
  Given a game $\Gamma$ over an interval $I$, a location $l$, and a
  cost function $h : I \to \real$, if $h(x) = \OptCost_{\Gamma}(l,x)$
  for every $x \in I$ then
  \[
    \OptCost_{\Gamma[\LGoal\cup l,h]}(l',x) = \OptCost_\Gamma(l',x),
    \]
  for every location $l' \in L$ and every clock valuation $x \in I$.
\end{lemma}

Lemma \ref{lem:one_location_computed} is a direct consequence of the
following Lemma, which characterizes the relation between the values of
optimal reachability cost of adjacent locations.

\begin{lemma}
  \label{lem:OptCost_opt_eq}
  Given $l \in \LMin$, let $h(x) = \min\set{\OptCost(l',x) \sq (l,l')
    \in E}$, we then have
  \[
    \OptCost(l) = \minC(h,\pi(l))
    \]
  If $l \in \LMax$ then, if we substitute $\max$ for $\min$ and $\maxC$
  for $\minC$, the same equality holds.
\end{lemma}

\paragraph{Algorithm.}
We will define a recursive function $\solverp$ that solves a
reachability-price game $\Gamma_I$, where $I = [b,e] \subseteq [0,1]$
and the automaton underlying $\Gamma$ is simple. Upon termination, the
function outputs $\OptCost_{\Gamma_I}$.

The algorithm works recursively, with respect to the set of non-urgent
locations. During each recursive call, it identifies a non-urgent
location that minimises the weight function. There are two cases to
consider, depending on the ownership of the location, however, both of
them are handled in a similar fashion. The algorithm modifies the game
$\Gamma$ to have one less non-urgent location. In case $l \in \LMax$,
we convert $l$ to be urgent, whereas if $l \in \LMin$, we convert $l$
to be a goal location that captures the following behaviour: once $l$
is reached in $\Gamma$, the minimizer spends all available time there.
The intuition behind this is as follows: if $l \in \LMax$, it is
unlikely that spending time in that location will be beneficial for
the maximizer. Likewise, when $l \in \LMin$, it is likely that it will
be beneficial for the minimizer to stay as long as possible.  There
are cases, however, when this intuition is incorrect, i.e., it is
beneficial, respectively, for the maximizer to wait, and for the
minimizer to move immediately. This necessitates the iterative
procedure, outlined in the following, employed during each recursive
call.

The working assumption is that $\OptCost$ is, locationwise, a cost
function.  During each recursive call, the algorithm iteratively
computes the result of the $\minC$ ($\maxC$) operator applied to the
minimum (maximum) of the location's successor's cost functions (that
are equal to $\OptCost$).  The iterative procedures in cases 2 and 3
of the algorithm compute the solution over a sequence of intervals,
proceeding from the left to the right of the time axis. They first
assume that the aforementioned intuition is correct (step 1), and then
identify the rightmost interval, over which it is not (step 2). The
next step is to adjust the solution over that interval (step 2 and
3). It remains to find the solution to the left of the found
interval. This is done in the subsequent iterations.

We now present the recursive algorithm $\solverp(\Gamma_I)$. There are
three cases to consider.

\textbf{First case:} $\nonurg(\Gamma_I) = \emptyset$. $\OptCost(l)$ is
a cost function (for every location $l$) and can be computed by
solving a finite game in polynomial time.  If $\afgoal$ has $p$ pieces
in total, then $\OptCost$ has at most $2p$ pieces \cite{BLMR06}. 

\textbf{Second case:} $\LMax \ni \lmin = \argmin\set{\pi(l) : l \in
  \nonurg(\Gamma)}$. 
In Case 2 of the algorithm, an iterative procedure is applied to
compute $\OptCost_\Gamma$ over the interval $I =[b,e]$; in each
iteration, the computation is restricted to the interval $[b,r]$, with
$r = e$ in the first iteration. First, in Step 1, a game $\Gamma'$
with one less non-urgent is constructed. We obtain $\Gamma'$ from
$\Gamma$ by making $\lmin$ an urgent location --- this captures the
intuition that, since $\lmin$ minimizes the weight function, it is
beneficial for the maximizer to leave $\lmin$ immediately. Second, in
Step 2, the procedure identifies the rightmost interval over which the
function $f = \OptCost_{\Gamma'}(\lmin)$, computed in Step 1, has an
affine piece with the slope strictly shallower than $-\pi(\lmin)$; the
affine piece and the interval are denoted by $f_i$ and $[b_i, e_i]$,
respectively. Third, in Step 2, a new game, $\Gamma''$, is
constructed; we are considering this game over the interval $[b_i,
  e_i]$. Like $\Gamma'$, the game $\Gamma''$ has one less non-urgent
location than the game $\Gamma$; it is obtained from $\Gamma$ by
turning $\lmin$ into a goal location with the cost function $h =
-\pi(\lmin)(r-x) + \OptCost_{\Gamma_{[r,e]}}(\lmin,r)$ assigned to
$\lmin$ --- this cost function captures the behaviour contrary to the
previously considered intuition, i.e., that, upon entering $\lmin$,
the maximizer spends all available time there. The game $\Gamma''$ is
used to adjust the solution, to account for states from which the
intuition that leaving $\lmin$ immediately is beneficial to the
maximizer is incorrect. The slope of $f_i$ is shallower than
$-\pi(\lmin)$, and since $\lmin$ minimises the weight function, this
means that $f_i$ is actually an affine piece of one of the cost
functions assigned to goal locations in the game~$\Gamma$. Finally, in
Step 3, $\OptCost_\Gamma$ over $[b_i,e]$ is being established. It is
equal to $\OptCost_{\Gamma'}$ over the interval $[e_i,r]$ and to
$\OptCost_{\Gamma''}$, over the interval $[b_i, e_i]$. The algorithm
then proceeds to the next iteration by setting $r = b_i$; the
iterative procedure is completed when $b_i = b$. The $\tcompose$
operation is used to assure consistency of solutions between
subsequent iterations.

The procedure is as follows:
\begin{enumerate}
  \item \label{alg:case2_step1} Assuming that we have computed
    $\OptCost_{\Gamma_{[r,e]}}$, for some $r \in I$, we set
    \[
      \Gamma_{[b,r]}' =
      \left(\tcompose\left(\Gamma_{\left[b,r\right]},\OptCost_{\Gamma_{\left[r,e\right]}}\right)\right)\left[\urg(\lmin):=1\right],
      \]
    and we compute $\OptCost_{\Gamma_{[b,r]}'} =
    \solverp\left(\Gamma_{[b,r]}'\right)$. Let $f =
    \piecewise{f_1,\ldots,f_k} = \OptCost_{\Gamma_{[b,r]}'}(\lmin)$.
  \item \label{alg:case2_step2} Let $i$ be the smallest natural number
    such that $a_i^f > -\pi(\lmin)$ and for all $j > i$ we have $a_j^f
    \leqslant -\pi(\lmin)$.  If $i > 0$, then we define $h : I_i^f \to
    \real$ as $-\pi(\lmin)(e_i^f - x) + f_i(e_i^f)$, and
    \[
     \Gamma_{\left[b_i^f,e_i^f\right]}'' = 
     \tcompose\left(\Gamma_{\left[b_i^f,e_i^f\right]}',\OptCost_{\Gamma_{\left[e_i^f,r\right]}'}\right)\left[\LGoal
       \cup \lmin,h\right],
      \]
    and compute $\OptCost_{\Gamma_{\left[b_i^f,e_i^f\right]}} =
    \solverp\left(\Gamma_{\left[b_i^f,e_i^f\right]}\right)$.
    \item \label{alg:case2_step3} 
      We set 
      \[
        \OptCost_{\Gamma_{\left[b_i^f,e\right]}} =
        \OptCost_{\Gamma_{\left[b_i^f,e_i^f\right]}''}\override
        \OptCost_{\Gamma_{\left[e_i^f,r\right]}'}\override\OptCost_{\Gamma_{\left[r,e\right]}}.
      \]
      If $i=0$ the $\Gamma''$ term is omitted.

      We set $r = b_i^f$. If $r \neq b$ then goto~\ref{alg:case2_step1},
      otherwise output $\OptCost_{\Gamma_I}$.
\end{enumerate}
We initialize the procedure by solving the game $\Gamma_{I}' =
\Gamma_{I}[\urg(\lmin):=1]$, and setting $r = e$. Observe that
$\OptCost_{\Gamma_{[e,e]}}= \OptCost_{\Gamma_{[e,e]}'}$.

\textbf{Third and last case:} $\LMin \ni \lmin = \argmin\set{\pi(l) :
  l \in \nonurg(\Gamma)}$. 
In Case 3 of the algorithm, an iterative procedure is applied to
compute $\OptCost_\Gamma$ over the interval $I =[b,e]$; in each
iteration, the computation is restricted to the interval $[b,r]$, with
$r = e$ in the first iteration. First, in Step 1, a game $\Gamma'$
with one less non-urgent is constructed. We obtain $\Gamma'$ from
$\Gamma$ by making $\lmin$ a goal location; the cost function, $h$,
assigned to $\lmin$ captures the following behavior, once $\lmin$ is
reached, the minimizer chooses to spend all available time
there. Second, in Step 2, the procedure identifies the rightmost
interval, over which the function $f$, computed in Step 1, is strictly
smaller than $h$ for at least one argument; the interval corresponds to
an affine segment of $f$, denoted by $f_i$. The argument, for which
the functions $f$ and $h$ are equal, is denoted by $x^* \in I_i$ ---
such an argument always exists as $f(r) = h(r)$, by definition. Third,
in Step 2, a new game, $\Gamma''$, is constructed. Like $\Gamma'$, the
game $\Gamma''$ is obtained from $\Gamma$ by turning $\lmin$ into a
goal location. In this case, however, the cost function assigned to
$\lmin$ is $f_i$, and the game is being considered over the interval
$[b_i, x^*]$ --- the cost function $f_i$ captures the intuition that
it is beneficial to leave $\lmin$ immediately. The game $\Gamma''$ is
used to adjust the solution, to account for states from which the
intuition that spending all available time in $\lmin$ is beneficial to
the minimizer is incorrect. The slope of $f_i$ is shallower than that
of $h$, which is equal to $-\pi(\lmin)$, and since $\lmin$ minimises
the weight function, this means that $f_i$ is actually an affine piece
of one of the cost functions assigned to goal locations in the
game~$\Gamma$. Finally, in Step 3, $\OptCost_\Gamma$ over $[b_i,e]$ is
being established. It is equal to $\OptCost_{\Gamma'}$ over the
interval $[x^*,r]$ and to $\OptCost_{\Gamma''}$, over the interval
$[b_i, x^*]$. The algorithm then proceeds to the next iteration by
setting $r = b_i$; the iterative procedure is completed when $b_i =
b$. The $\tcompose$ operation is used to assure consistency of
solutions between subsequent iterations.

The procedure is as follows:
\begin{enumerate}
  \item \label{alg:case3_step1} Assuming that we have computed
    $\OptCost_{\Gamma_{[r,e]}}$, for some $r \in I$, let $h : [b,r]
    \to\real$ be defined as $h(x) =
    -\pi(\lmin)(r-x)+\OptCost_{\Gamma_{[r,e]}}(\lmin,r)$, we set
    \[
      \Gamma_{[b,r]}' =
      \left(\tcompose\left(\Gamma_{[b,r]},\OptCost_{\Gamma_{[r,e]}}\right)\right)\left[\LGoal
        \cup \lmin,h\right]
      \]
    and compute $\OptCost_{\Gamma_{[b,r]}'} =
    \solverp\left(\Gamma_{[b,r]}'\right)$.  We define $f =
    \piecewise{f_1,\ldots,f_k}$ as $\min\set{
      \OptCost_{\Gamma_{[b,r]}'}(l)\sq (\lmin,l) \in E}$.
  \item \label{alg:case3_step2} Let $i$ be the smallest natural number
    such that for all $j > i$ we have $f(x) \geqslant h(x)$ over
    $[b_j^f, e_j^f]$. If $i > 0$ let $x^*$ denote the solution of $f(x) =
    h(x)$ (over $[b_i^f,e_i^f]$). We then set
    \[
      \Gamma_{\left[b_i^f,x^*\right]}'' =
      \tcompose(\Gamma_{\left[b_i^f,x^*\right]}',\OptCost_{\Gamma_{\left[x^*,r\right]}'})\left[\LGoal
        \cup \lmin,f_i\right]
      \]
    and compute $\OptCost_{\Gamma_{[b_i^f,x^*]}''} =
    \solverp\left(\Gamma_{[b_i^f,x^*]}''\right)$.
    \item \label{alg:case3_step3} We set 
      \[
        \OptCost_{\Gamma_{\left[b_i^f,e\right]}} =
        \OptCost_{\Gamma_{\left[b_i^f,x^*\right]}''}\override\OptCost_{\Gamma_{\left[x^*,r\right]}'}\override\OptCost_{\Gamma_{\left[r,e\right]}}.
        \]
      If $i=0$ then the $\Gamma''$ term is omitted.

    We set $r = b_i$. If $r \neq b$ then goto~\ref{alg:case3_step1},
    otherwise output $\OptCost_{\Gamma_I}$.
\end{enumerate}
We initialize the procedure by solving the game $\Gamma_{[e,e]}$, and
setting $r = e$.  Observe that this can be done in polynomial time.

The following example provides the intuition behind
the iterative procedure employed during each recursive call of the
algorithm.

\begin{example}
  Fig.~\ref{fig:ta_min_iteration_example} shows how the iterative
  procedure in Case 3 of the algorithm works to compute
  $\OptCost_\Gamma$ over the interval $I = [b,e]$. In diagram a) we
  can see that $\OptCost_\Gamma$ has been computed over the interval
  $[r,e]$. The function $h$ denotes the cost function assigned to
  $\lmin$ in $\Gamma'$ and the dashed line denotes the function $f$,
  as defined in Step~1 of Case~3. One can see that the interval
  $[b_i,e_i]$ and $x^*$, identified in Step 2 of Case 3 of the
  algorithm, are such that: over the interval $[x^*,r]$ the intuition,
  which indicates that the minimizer should spend all the available
  time in $\lmin$, is correct; and that over the interval $[b_i,x^*]$,
  this intuition is not correct, i.e., it is beneficial for the
  minimizer to leave $\lmin$ immediately --- the dashed and bold
  segment of $f_i$ denotes the cost function assigned to $\lmin$ in
  $\Gamma''$. In diagram b) we can see the next iteration of the
  algorithm. $\OptCost_\Gamma$ has been computed over $[r'=b_i,e]$;
  this iteration follows the same steps as the previous one. Note
  that, over the interval $[b_i,r]$, $\OptCost_\Gamma(l)$ is equal to
  $h$, over the interval $[x^*,r]$ and to $f_i$, over the interval
  $[b_i,x^*]$, as defined in Step 3 of Case 3 of the algorithm.\qed
\end{example}

\begin{figure}
  \begin{center}
    \begin{tabular}{c c}
      \includegraphics[scale=0.4]{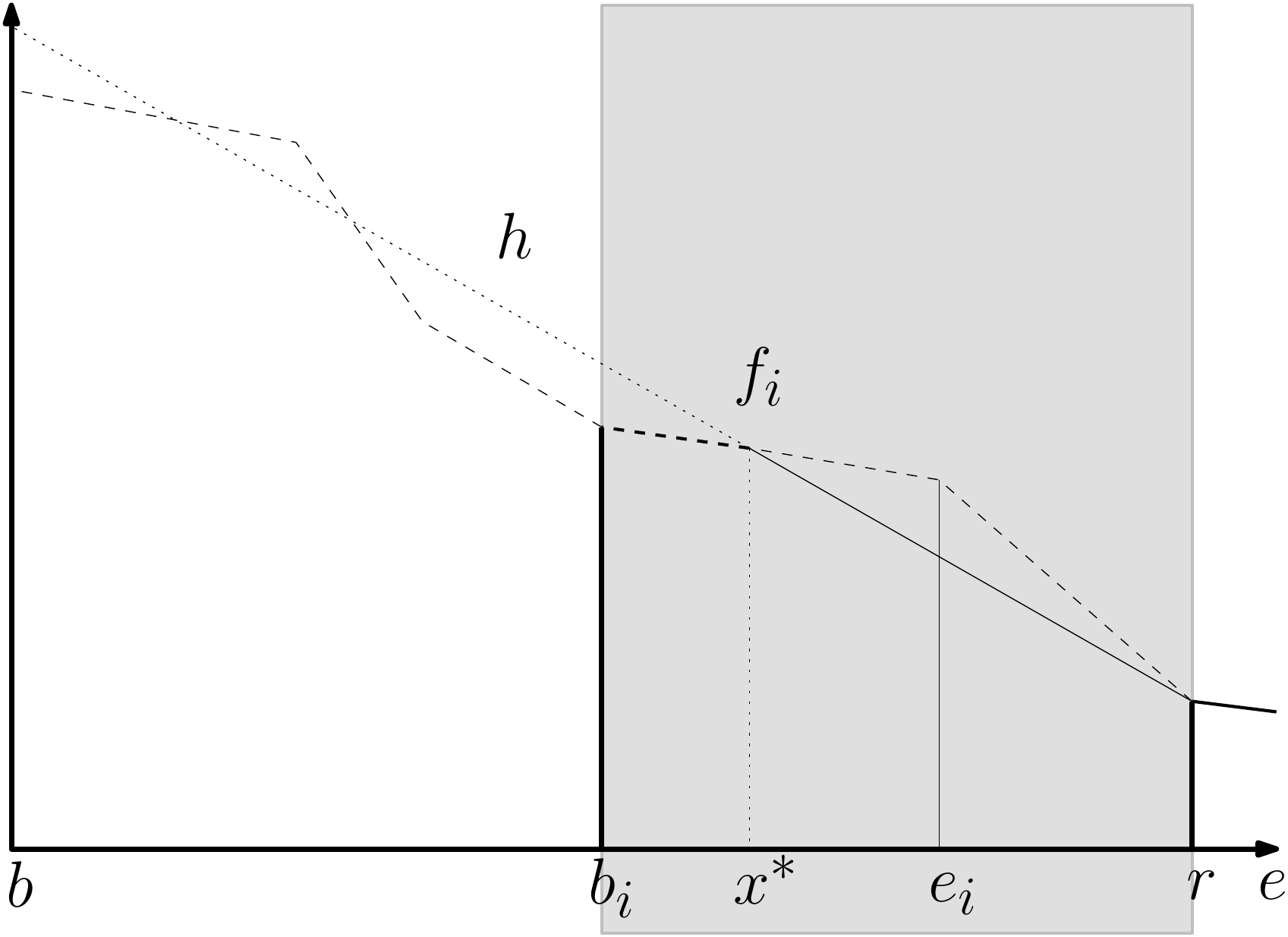} &
      \includegraphics[scale=0.4]{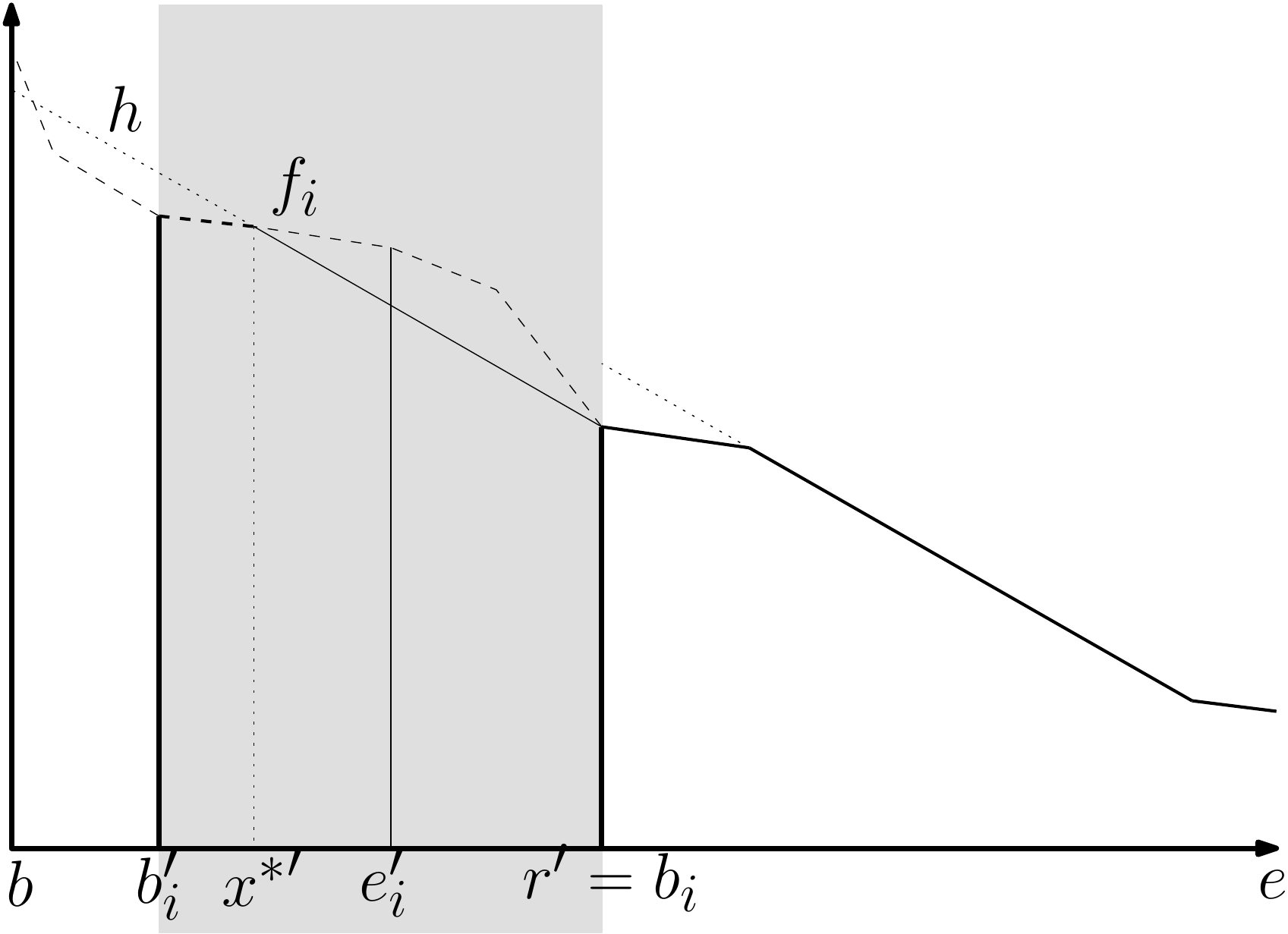} \\
      a) & b)
    \end{tabular}
    \caption{Iterative computation of $\OptCost_\Gamma(\lmin)$ over
      the interval $I = [b,e]$, where $\lmin \in \LMin$. Diagrams a)
      and b) depict two subsequent iterations; the gray rectangle
      indicates the subinterval for which the $\OptCost_{\Gamma}(\lmin)$ function is
      being computed during the given iteration.}
      \label{fig:ta_min_iteration_example}
  \end{center}
\end{figure}

\paragraph{Correctness and complexity.}
We show that the procedure $\solverp$ is correct, i.e., that if it
terminates, the output is in fact the $\OptCost$ function, and that it
indeed terminates. We will also show, that there is an exponential
upper bound on the running time of $\solverp$.
The main result of this paper is as follows:

\begin{theorem}
  Given a reachability-price game $\Gamma$, the function
  $\solverp(\Gamma)$ terminates and outputs the function
  $\OptCost_\Gamma$.
\end{theorem}

We will prove the Theorem in two steps. First, we prove that if the
iterative procedure in cases 2 and 3 terminates, it computes
$\OptCost$. Second, we show that it always terminates.

\begin{theorem}
  \label{theo:partially_correct}
  Given a reachability-price game $\Gamma$, if $\solverp(\Gamma)$
  terminates, it outputs the function $\OptCost_\Gamma$.
\end{theorem}

\begin{proof}
 The proof is inductive. Fix $\Gamma$, and let $l$ be the non-urgent
 location that minimizes $\pi(l)$. Assume that $\Gamma$ has $n+1$
 non-urgent locations, and that Theorem \ref{theo:partially_correct}
 holds for every game $\Gamma'$ that has at most $n$ non-urgent
 locations.

 If there are no non-urgent locations, then computing
 $\OptCost_{\Gamma}$ amounts to solving a reachability-price game on a
 finite graph.  It remains to prove the inductive step; there are two
 cases to consider. The first case, when $l \in \LMin$ and the second
 case when $l \in \LMax$. The proofs of these two cases follows from
 Lemmas \ref{lem:case_2_partially_correct} and
 \ref{lem:case_3_partially_correct}.\qed
\end{proof}

\begin{lemma}
  \label{lem:case_2_partially_correct}
  Given a reachability-price game $\Gamma$ with the price-rate
  minimizing location in $\LMax$, if Case 2 of $\solverp$ terminates,
  it outputs $\OptCost_\Gamma$.
\end{lemma}

\begin{proof}
  Case 2 is handled in the same way as in Bouyer's et al. algorithm
  \cite{BLMR06}.\qed
\end{proof}

\begin{lemma}
  \label{lem:case_3_partially_correct}
  Given a reachability-price game $\Gamma$ with the price-rate
  minimizing location in $\LMin$, if Case 3 of $\solverp$ terminates,
  it outputs $\OptCost_\Gamma$.
\end{lemma}

\begin{proof}
  Without loss of generality we assume that we are dealing with a
  single interval $I$. Let $\lmin$, $\Gamma'$, $\Gamma''$, $f$, and
  $x^* \in I$ be defined as in the Case 3 of the algorithm.

  We will show that $\OptCost_{\Gamma'}(\lmin) =
  \OptCost_{\Gamma}(\lmin)$ over the interval $[x^*,e]$ and that
  $\OptCost_{\Gamma''}(\lmin) = \OptCost_{\Gamma}(\lmin)$ over the
  interval $[b,x^*]$. This, together with Lemmas
  \ref{lem:timed_composition} and \ref{lem:one_location_computed}
  enables us to establish that the procedure for computing
  $\OptCost_\Gamma$ over $I$ in Case 3 is correct. By the inductive
  hypothesis, we can solve $\Gamma'$ and $\Gamma''$ --- these games
  have one less non-urgent location than $\Gamma$. Additionally, in
  games $\Gamma'$ and $\Gamma''$ we restrict the minimizer so we
  immediately have:
  \[
    \OptCost_{\Gamma^{'/''}}(l)(x) \geqslant \OptCost_{\Gamma}(l)(x),
    \]
  for every $l \in L$ and every $x \in [x^*,e]/[b,x^*]$. By definition
  of $\Gamma'$, there is an equality for $x = e$. To complete the
  proof, we need to show the reverse inequality.

  There are two cases to consider. We start with the first case, i.e.,
  we will show that $\OptCost_{\Gamma'}(\lmin) \leqslant
  \OptCost_{\Gamma}(\lmin)$ over $[x^*,e]$.

  Fix $\varepsilon > 0$ and a strategy $\mu \in \SigmaMin$ (note that
  $\Gamma$ and $\Gamma'$ admit the same sets of strategies).  We take
  $\chi_\varepsilon \in \SigmaMax$ that is $\varepsilon$-optimal for
  $\mu$ in $\Gamma'$. For every $s \in \eset{\lmin}\times [x^*,e]$ let
  $\omega$ denote the unique run $\Run(s,\mu,\chi_\varepsilon)$, and
  lets assume it visits $\lmin$ exactly $m$ times, after transitions
  $i_1,\ldots,i_m$. We assume, without loss of generality, that after
  every such transition, a continuous one is taken. We then have:
  \begin{eqnarray*}
    \Cost_{\Gamma}(\omega) &=& \Price(\omega_{i_m}) + \lambda_{i_m +
      1}\cdot\pi(\lmin) + \Cost_{\Gamma}(\Run(\omega_{i_m
      +1},\mu,\chi_\varepsilon))\\ &\geqslant& \Price(\omega_{i_m}) +
    \lambda_{i_m + 1}\cdot\pi(\lmin) + \OptCost_{\Gamma'}(s_{i_m + 1}) -
    \varepsilon\\ &\geqslant& \Price(\omega_{i_1}) + (x_{i_m} +
    \lambda_{i_m + 1}-x_{i_1})\cdot\pi(\lmin) + \OptCost_\Gamma(s_{i_m
      + 1}) - \varepsilon\\ &\geqslant& \Price(\omega_{i_1}) +
    (e-x_{i_1})\cdot\pi(\lmin) + \OptCost_{\Gamma}((\lmin,e)) -
    \varepsilon\\
    &=&
    \OptCost_{\Gamma'}(s) - \varepsilon,
  \end{eqnarray*}
  where $s_i = (l_i,x_i)$. The first inequality holds because
  $\chi_\varepsilon$ is $\varepsilon$-optimal for $\mu$ in~$\Gamma'$,
  and because the suffix of $\omega$ starting in $s_{i_m}$ does not
  visit $\lmin$ but for the first transition, which comes at a price
  zero. The second inequality holds because $\lmin$ minimizes
  $\pi(\lmin)$. The third inequality follows from the definition of
  $\OptCost_{\Gamma'}(\lmin)$. Finally, $\Length(\omega_{i_1}) = 0$,
  hence $\Cost(\omega_{i_1}) =0$.

  We now proceed to
  the second case, i.e., that
  $\OptCost_{\Gamma''}(\lmin) \leqslant \OptCost_{\Gamma}(\lmin)$ over
  $[b,x^*]$.  The fact that $\OptCost_{\Gamma''}(\lmin) \leqslant
  \OptCost_{\Gamma'}(\lmin)$ over $[b,x^*]$ implies that the slope of $f$
  is shallower than $-\pi(\lmin)$.

  Fix $\varepsilon > 0$ and $\mu \in \SigmaMin$ in $\Gamma$.  Let
  $\chi_\varepsilon \in \SigmaMax$ be a strategy that is
  $\varepsilon$-optimal for $\mu$ in $\Gamma''$ (the sets of
  strategies in $\Gamma$ and $\Gamma''$ are equal). For every $s \in
  \eset{l}\times [b,x^*]$, let $\omega$ denote the unique run
  $\Run(s,\mu,\chi_\varepsilon)$, which pays $m$ visits to $\lmin$
  (the notation and assumptions are the same as in the first case). We
  then have:
  \begin{eqnarray*}
    \Cost_\Gamma(\omega) & = &\Price(\omega_{i_m}) + \lambda_{i_m + 1}
    \cdot\pi(\lmin) + \Cost_\Gamma(\Run(\omega_{i_m
      +1},\mu,\chi_{\varepsilon'}))\\ & \geqslant
    &\Price(\omega_{i_m}) + \lambda_{i_m +1} \cdot\pi(\lmin) +
    \OptCost_{\Gamma''}(s_{i_m +1},\mu,\chi_{\varepsilon'})) -
    \varepsilon\\ & \geqslant & \Price(\omega_{i_1}) + (x_{i_m}+
    \lambda_{i_m +1}-x_{i_1}) \cdot\pi(\lmin) +
    \OptCost_{\Gamma''}(s_{i_m +1}) - \varepsilon\\ &\geqslant&
    \Price(\omega_{i_1}) + \OptCost_{\Gamma''}(s_{i_1})
    -\varepsilon\\ &=& \OptCost_{\Gamma''}(s_{i_1}) -\varepsilon.
  \end{eqnarray*}
  The first inequality holds because $\omega$ does not feature
  transitions ending in $\lmin$, modulo its prefix $\omega_{i_m +1}$,
  and because $\chi_{\varepsilon}$ is $\varepsilon$-optimal for $\mu$
  in $\Gamma''$. The second inequality holds because $\lmin$ minimizes
  $\pi(\lmin)$. The final inequality holds because the slope of $f$ is
  shallower than $-\pi(\lmin)$ over $[b,x^*]$.  This finishes the
  proof of the theorem.  \qed
\end{proof}

We have proved that the algorithm is partially correct. It remains to
prove its total correctness, i.e., that it terminates. 

\begin{theorem}
  \label{lem:iteration_termination}
  The algorithm $\solverp$ terminates.
\end{theorem}

\begin{proof}
  To prove termination of the algorithm, we need to prove the
  termination of the iterative procedures from cases 2 and 3 of the
  algorithm.

  Each of the two cases is different, however, they have one thing in
  common. In each iteration, in both cases, an interval with slope
  shallower than $-\pi(\lmin)$ is processed. Since $\lmin$ minimizes
  $\pi(\lmin)$, and by Lemmas \ref{lemma:minC_piecewise_affine} and
  \ref{lem:OptCost_opt_eq}, this interval corresponds to a segment of
  a goal cost function. We argue that the number of iterations in each
  case is bounded by the number of all the possible intersections of
  the cost functions assigned to goal locations, which is finite.

  More precisely, let $\Gamma'$, $I$, $f$ and $i$ be defined as in Case
  2 of the algorithm. The slope of $f_i$ over $I_i$ is shallower than
  $-\pi(\lmin)$, so by Lemmas \ref{lemma:minC_piecewise_affine} and
  \ref{lem:OptCost_opt_eq}, $f_i$ coincides with some cost function
  over $I_i$ --- denoted by $g$. If $i > 1$, then there are two
  possibilities, either $b_i$ coincides with an intersection of cost
  functions from two different goal locations, or otherwise $f_{i-1}$
  has the slope equal to $-\pi(l)$ for some non-goal location~$l$.
  
  In the first case, the iteration must have moved past one of the finitely
  many intersection points. In the second case, we need to argue that
  if the procedure once again encounters the same affine piece of
  $g$ (but over a different interval), then it must have also passed
  at least one of the finitely many intersection points. Let $I' =
  [b',e'] \subseteq [b,b_i)$ denote the interval over which the
    procedure encounters $g$ again. Assume that the procedure did not
    pass any intersection points before $I'$. This implies that
    $\OptCost_{\Gamma'}(\lmin) \geqslant g$ over $[e',b_i]$, and hence
    $\OptCost_{\Gamma'}(\lmin)$ must contain an affine piece that has a
    slope shallower than $g$ over $[e',z]$, for some $z \in (e',b_i]$.
  However, such a piece coincides with a goal cost function, so an
  intersection point must have been passed.

  So far we have shown termination of the iterative procedure in Case
  2. It remains to show the same for Case 3. Let $\Gamma'$, $f$, $i$,
  and $x^*$ be defined as in Case~3 of the algorithm. We argue that
  each affine segment of a goal cost function is processed only once.
  If $i > 1$, then $f_{i-1}$ either coincides with a different piece
  of a goal cost function, which means that we have passed one of the
  finitely many intersection points and the $f_i$ segment has been
  processed, or its slope is equal to $-\pi(l)$, for some non-goal
  location $l$. In the latter case, we have that
  $\OptCost_{\Gamma'}(\lmin)$ has a slope steeper than $f_i$, and
  hence, it is strictly greater than $f_i$ over $[b,b_i)$. This means
    that in the subsequent iterations, if $f_i$ is to be encountered,
    a piece with a slope smaller than that of $f_i$ needs to occur,
    but this means that an intersection point has been processed.

    We have shown that in each step of the algorithm the iterative
    procedure of Cases 2 and 3 terminates, and hence, the algorithm
    terminates.
    \qed
\end{proof}

We have proved that the procedure $\solverp$ is correct. The question
that remains, is its complexity. We have the following result.

\begin{theorem}
  The algorithm $\solverp$ is in EXPTIME.
\end{theorem}

\begin{proof}
  Given an automaton $\Aa$, let $n$ denote the number of non-urgent
  locations and $p$ the total number of pieces in $\afgoal$.  The
  complexity of computing the solution, using $\solverp$, depends on
  the number of pieces that constitute $\OptCost_\Gamma$. Let $N(n,p)$
  denote the upper bound on the number of pieces in $\OptCost_\Gamma$.
  We now construct a recursion to characterise $N(n,p)$.

  If $n = 0$, then $N(n,p) = 2p$ \cite{BLMR06}. If $n > 0$, both cases
  take $p$ (as argued in the proof of Theorem
  \ref{lem:iteration_termination}) iterations, and each requires
  solving two games with the solution complexity equal to $N(n-1,p+n)$
  and $N(n-2,p+n-1)$. We can assume that $p > n$ is the case of real
  interest, so we have $N(n,p) \leqslant 2p N(n-1, 2p)$. It can be
  easily verified that:
    $N(n,p) \leqslant 2^{\frac{(n+1)(n+2)}{2}}p^{n+1}$.
  This establishes that $\solverp$ is indeed in EXPTIME.\qed
\end{proof}

\paragraph{Discussion.}
We now briefly compare our algorithm with that of Bouyer et al. The
3EXPTIME algorithm, introduced by Bouyer et al.~\cite{BLMR06}, differs
from the one presented in this paper in the way the minimizer
locations are handled. As was explained above, the algorithm presented
in this paper uses an iterative procedure, similar in spirit to that
used for handling maximizer locations. The 3EXPTIME algorithm, on the
other hand, exploited the following observation: if the location, $l$,
which minimizes the weight function, is visited several times, then
the minimizer would not be worse off, if, upon the first visit, she
had waited the whole time that passes between the first and last visit
--- this is valid because all other locations have a higher value of
the weight function.  This intuition is formally captured by the
algorithm in the following way. Two copies of the original automaton
are created, with the only difference that in both of them $l$ becomes
a goal location, and hence both automata have one less non-urgent
location --- this duplication introduces an exponential blowup in
complexity. The first automaton captures the behavior before $l$ is
entered for the first time, whereas the second copy captures the
behaviour afterwords. In the second copy, $l$ is transformed into a
goal location with a cost function equivalent to positive infinity ---
this captures the intuition, that it is sufficient to visit $l$ only
once. The algorithm first computes $\OptCost$ for the second copy,
then, using that result, computes $\OptCost$, for all states having
$l$ as the location (this is in fact a game with a single non-urgent
location), and finally computes $\OptCost$ for the first copy, with
$\OptCost$, computed in the previous step, being assigned as a cost
function to $l$. $\OptCost$ computed for the first copy is the sought
solution. The second exponential blowup originated from the
construction that allowed to assume that the clock value is bound by
$1$. The construction used by Bouyer et al.~\cite{BLMR06} used
locations to encode the integer part of the clock value, and the clock
itself captured only the fractional part of the clock value --- this
yielded an exponential blowup, as there had to be a copy of the
original location for every integer value between $0$ and the value of
the largest constant provided in the definition of the automaton
(recall, that constants are encoded in binary). Our construction,
presented in Sec.~\ref{sec:preliminaries}, avoids this blowup.

%
%
\bibliographystyle{eptcs}

\bibliography{article_bibliography}

\end{document}